\documentclass[jmp,preprint]{revtex4-1}
\draft

\usepackage{amsfonts, amsmath, amssymb, amsthm, graphicx,mathtools }
\usepackage{enumerate}
\usepackage{natbib}
\newtheorem{theorem}{Theorem}
\newtheorem*{theorem*}{Theorem}

\newtheorem{Lemma}[theorem]{Lemma}

\newtheorem{corollary}[theorem]{Corollary}

\newtheorem{definition}[theorem]{Definition}
\newtheorem{example}[theorem]{Example}

\usepackage{bbold}

\begin{document}
\title{Matrix Product States approach to non-Markovian processes} 
\author{Beno\^it Descamps}%

\email{benoit.descamps@univie.ac.at}
\affiliation{Faculty of Physics, University of Vienna, Austria}%

\date{\today}

\begin{abstract}
A matrix product states approach to non-Markovian, classical and quantum processes is discussed. In  the classical case, the Radon-Nikodym derivative of all processes can be embedded into quantum measurement procedure. In the both cases, quantum and classical, the master equation can be derived from a projecting a quantum Markovian process onto a lower dimensional subspace.
\end{abstract}

\maketitle

\section{Introduction and summary}
Stochastic processes play a major role in any fields of science. Whether it is physics, biology, chemistry, finance, etc..., the effective dynamics of  macro-variables are derived from a couplings to some random observables. 

Consider a process $(X_j)_{j\in\mathbb{N}}$ of which $X_1$, $X_2$, ..., $X_n$ could represent observations sampled randomly
from the same population at time points i = 1, 2, .... When constructing models some underlying assumption of the mechanics has to be made.
The most famous example are the  Autoregressive-Moving-Average model $ARMA(p,q)$ processes ,
$$X_k = c+\sum_{j_1=1}^p \alpha_{j_1} X_{k-j_1}+\sum_{j_2=1}^q \theta_{j_2}W_{k-j_2}$$
with $W_k$ are independently normally distributed  $N(0,1)$. Notice that the transition probabilities only depend on finite variables. Hence,
these processes provide a useful method to describe the statistics of short-memory processes. However, as it is well known, there exists processes with long memory. Some of those have been described by Beno\^it Mandelbrot and are known as fractional Brownian motion.

Clearly, for all these processes an analytically approachable form had to be chosen in order to model the statistical properties desired.
In the end, given some data $(X_j)_{j\in\mathbb{N}}$, we wish to find the model with the right correlations and higher momenta. These are fully described by the characteristic function of the process,
$$f(u_1,u_2...)=E(\exp(iu_1 X_1+iu_2 X_2+...))$$
Therefore the underlying mechanics of should not be described on the level of the observations, but rather on the level of the characteristic function or the measure of the process.

In \cite{sMPSTemme},  Stochastic Matrix product States (sMPS) were introduced as ansatz for studying non-equilibrium states in statistical mechanics.
As shown further on, similarly to Matrix Product States (MPS) being fixed points of the Density Renormalization Group, sMPS should be seen as fixed points of Metropolis Monte-Carlo Sampling. Hence, these states represent the joint distribution of a statistical (non-Markovian) process.

The paper is ordered as follows. In the first section, we present an introduction to Matrix Product States (MPS), their continuum version (cMPS), Quantum Measurements, and Radon-Nikodym derivatives. In the second section, we show the connection with the so-called sMPS and how it can be used to described the non-Markovian processes and their Master equations with some examples.
In the last, section some miscellaneous applications can be found

\subsection{Matrix Product States and their Continuum versions}
Matrix Product States form a class of finitely correlated states used to study the low energy spectrum of local Hamiltonians in quantum spin chain.
For a chain $\otimes_{j=1}^N\mathbb{C}^d$ of size $N$, they are given as follows,

$$|\psi \{A^{(j)}[k]\}\rangle =\sum_{i_1,...,i_N}\operatorname{Tr}\left(A^{(i_1)}[1]...A^{(i_N)}[N] \right)(\sigma^+[1])^{i_1}...(\sigma^+[N])^{i_N} |0...0\rangle=\operatorname{Tr}_V(Z_N)|0...0\rangle$$
with matrix $A^{(j)}[k]\in \mathcal{M}_D$ of bound dimension $D$. The partial trace in the last term is in the virtual space. The following normalization condition can be chosen,
$$\sum_{j}A^j[k]A^j[k]^\dagger =\mathbb{1}$$
For the case $d=2$, the operator $Z_N$ can be rewritten as a solution of the recursion equation,
$$Z_N=Z_{N-1}+ Z_{N-1} \left(A_0[N]-\mathbb{1}\right)+Z_{N-1}A^1[N]\otimes \sigma^+$$
Using notation $\Delta n =1$, $\Delta Z_n = Z_{n}-Z_{n-1}$, $\Delta B^\dagger_n =\frac{1}{\sqrt{\Delta n}}\sigma^+[n]$, this equation is equivalent to,
$$Z_N=\mathbb{1}+\sum_{j=1}^N \Delta Z_n,~~~\Delta Z_n = Z_{N-1} Q[n]\Delta n+Z_{N-1}R[n]\otimes \Delta B^\dagger_n,~ Q[n]=\frac{\mathbb{1}-A_0[N]}{\Delta_n},~~R[n]=\sqrt{\Delta_n}A^1[n]$$
This is a discrete quantum stochastic differential equation (QSDE). Clearly by identifying $\Delta n \leftrightarrow dt$ and $\Delta B^\dagger[n]\leftrightarrow dB^\dagger(t) $, the continuum version of matrix product states can be
Its continuum version is ,
$$|\psi \{Q[x], R[x]\}\rangle =\operatorname{Tr}\left(\exp\left(\int_0^L Q(x) ds + R(x)\otimes dB^\dagger(x)\right)\right) |\Omega\rangle =\operatorname{Tr}_V(Z_N)|0...0\rangle$$
which satisfies again the QSDE,
$$Z_t=\mathbb{1}+\int_0^t  dZ_s,~~~dZ_t = Z_t Q[t]dt+Z_{t}R[t]\otimes dB^\dagger[t]$$
A wider variety of solutions of QSDE have been studied by Hudson and Parthasarathy \cite{HUDSON}. We do not continue on this topic.
\subsection{Quantum Measurement Theory}
Denote $(Y,\mathcal{B}(Y))$  a Borel space representing the output events of some quantum measurement. Define a random variable $X:Y\to \mathbb{R}$ which represent the value of the outputs. Measurements are described by operators $A^{(y)}$ labeled by $y\in Y$. It is then postulated in quantum statistics that the probability of an event in $E\in \mathcal{B}(Y)$ is given,
$$P(X(y), y\in E)=\int_{y\in E}\operatorname{Tr}\left(\rho A^{(y)}A^{(y)\dagger}\right)$$
The conditional state $\rho$ resulting from the measurement is then given by,
$$\rho_E =\int_{y\in E}\frac{A^{(y)\dagger}\rho A^{(y)}}{\operatorname{Tr}\left(\rho A^{(y)}A^{(y)\dagger}\right)}$$

From this we can consider a sequence of measurement at different times $t_1,\dots,t_N$ with output values given by $X_{t_1},\dots,X_{t_n}$. The probability of a certain chain of outcomes $x_1,\dots,x_N$ at times $t_1,\dots,t_N$ is therefore given by, 
\begin{align*}
P(X_{t_1}=x_1,\dots,X_{t_n}=x_n)=&P(X_{t_n}=x_n|X_{t_{n-1}}=x_{t_{n-1}},\dots,X_{t_{1}}=x_{t_{1}})\dots\\
& P(X_{t_2}=x_2|X_{t_{1}}=x_{t_{1}})P(X_{t_1}=x_{t_1})\\
&=\langle\rho| \left(A^{(x_1)}\otimes \overline{A}^{(x_1)}\right)\dots \left(A^{(x_N)}\otimes \overline{A}^{(x_N)}\right)|I\rangle
\end{align*}
From another perspective, this measurement process can be seen as the evolution of a state in a cavity interacting with an electromagnetic field. The outputs are the values given by the detector in contact with the field. An introduction of to this topic can be found in \cite{HANDEL}.
In \cite{HOLEVO}, Holevo studied the representation of continuous measurement. Clearly not any arbitrary operator can be considered when considering a continuum limit.
In order to derive the representation, Holevo studied the characteristic function of the joint probability distribution of the sequence of output.
The idea is that similarly to infinitely divisible processes in statistics, also known as Levy processes, to look at the characteristic function.
\begin{align*}
\sum_{x_1,...,x_N}\exp(i\lambda(x_1+...+x_n))P(X_{t_1}=x_1,\dots,X_{t_n}=x_n)=\operatorname{Tr}\left(\rho\phi(\lambda)^n[\mathbb{1}]\right)=\operatorname{Tr}\left(\rho\exp\left(n\mathcal{L}(\lambda)\right)[\mathbb{1}]\right)\\
\phi(\lambda)[.]= \sum_{x}\exp(i\lambda x)A^{(y)}[.]A^{(y)\dagger},~~\mathcal{L}(\lambda)[.]=\operatorname{Id}-\frac{1}{n}\left(\Phi(n)-n\operatorname{Id}\right)[.]
\end{align*}
Clearly the limit $n\to \infty$ does not always exists, however when it does, the resulting generator $\mathcal{L}(\lambda)$ is given by a non-commutative version of the Levy-Khintchin representation theorem. One set of representations, we are interested in are of the form,
$$\mathcal{L}(\lambda)=\mathcal{L}_0[.]+\mathcal{L}_1(\lambda)[.] $$
with,
\begin{equation}
\label{Lindblad}
\mathcal{L}_0[.]=Q[.]+[.]Q^\dagger+\sum_j R_j[.]R^{\dagger}_j,~~Q=iH+\frac{1}{2}\sum_j R_jR_j^{\dagger}
\end{equation}
and,
$$\mathcal{L}_1(\lambda)[.]=im\lambda[.] +\sigma^2\left(R[.]R^\dagger-\frac{1}{2}\{RR^\dagger,.\}+ i\lambda (R[.]+[.]R^\dagger)-\frac{1}{2}\lambda^2[.]\right)$$
Consider the following two examples,
\begin{example}
For,
$$dZ^{(1)}(t)=Z^{(1)}(t)\left(\mathcal{L}_0[.]dt+ \sigma\frac{R[.]+[.]R^\dagger}{2} dB^x_t-i\sigma\frac{R[.]-[.]R^\dagger}{2}dB^y_t +\frac{1}{2}\sigma m(dB^x_t+idB^y_t)\right)$$
with solution
 \begin{align*}
 Z^{(1)}_t=&\exp\left(t\mathcal{L}_0[.]dt -\frac{1}{2}\sigma^ t\left( \left(\frac{R[.]+[.]R^\dagger}{2}+m\right)^2 -\left(\frac{R[.]-[.]R^\dagger}{2}-m\right)^2 \right)+\right.\\
 &\left. \sigma\frac{R[.]+[.]R^\dagger}{2} B^x_t-i\sigma\frac{R[.]-[.]R^\dagger}{2}B^y_t +\sigma m(B^x_t+iB^y_t)\right)
 \end{align*}
 
 where $(B^x_t,B^y_t)$ is a two-dimensional Brownian motion 
 We can verify using Ito-calculus,
\begin{align*}
E(\operatorname{Tr}_V(\rho Z^{(1)}_t)\exp(i\lambda (B^x_t+i B^y_t)))=\operatorname{Tr}\left(\rho \exp(t\mathcal{L}_0+t\mathcal{L}_1(\lambda))[\mathbb{1}]\right),\\
\end{align*}
\end{example}
We find a similar construction using 1-dimensional Brownian motion,
\begin{example}
For,
 \begin{equation}
 \label{mainSDE}
 dZ^{(2)}_t=Z^{(2)}_t\left(\mathcal{L}_0[.]dt +\sigma\left(R[.]+[.]R^\dagger\right) dB_t+ m dB_t\right)
 \end{equation}
 
 $$Z^{(2)}_t=\exp\left(t\mathcal{L}_0[.]t-\frac{1}{2}\sigma^2t\left(R[.]+[.]R^\dagger+m\right)^2+ \sigma\left(R[.]+[.]R^\dagger+m\right)B_t \right)$$
 where $(B^x_t,B^y_t)$ is a two-dimensional Brownian motion 
 We can verify using Ito-calculus,
\begin{align*}
E(\operatorname{Tr}_V(\rho Z^{(2)}_t)\exp(i\lambda B_t))=\operatorname{Tr}\left(\rho \exp(t\mathcal{L}_0+t\mathcal{L}_1(\lambda))[\mathbb{1}]\right),\\
\end{align*}
\end{example}

As we see in the next section, stochastic matrix product states represent as explained here a sequential measurement of a cavity coupled to an electromagnetic field. The examples provided here are the continuum limits of such states.

\subsection{Radon-Nikodym derivatives and Girsanov's theorem}
Let $(\Omega,\mathcal{F},\mathbb{P})$ be a probability space, let $\tilde{\mathbb{P}}$ be another probability measure on $(\Omega,\mathcal{F})$, that is equivalent to $\tilde{\mathcal{P}}$. Then there exist  an almost surely positive random variable $Z$ be which satisfies $\mathbb{E} Z=1$ and for $A\in \mathcal{F}$,
$$\tilde{\mathbb{P}}(A)=\int_A Z(\omega)d\mathbb{P}(\omega)$$

The statement described above is the Radon-Nikodym theorem. The random variable $Z$, called the Radon-Nikodym derivative, and allow us to change a measure.
Suppose we have a filtration $\mathcal{F}(t)$, define for $0\leq t\leq T$, where T is a fixed final time. Then we can define the so-called Radon-Nikodym derivative process,
$$Z(t)=\mathbb{E}[Z|\mathcal{F}(t)]$$
The main application of the change of measure, is that it allows to map a Brownian motion onto another one with a drift. This is procedure is known under the name Girsanov's theorem.
\begin{theorem}[Girsanov]
Let $B(t)$, $0\leq t \leq T$, be a Brownian motion on a probability space $(\Omega,\mathcal{F},\mathbb{P})$, let $\mathcal{F}(t)$, be the filtration generated by this Brownian motion. Let $\Theta(t)$, $0\leq t \leq T$, be an adapted process, define,
\begin{align*}
Z(t)&=\exp\left(-\int_0^t \theta(u)dB(u)-\frac{1}{2}\int_0^t \Theta^2(u)du\right)\\
\tilde{B}(t)=B(t)+\int_0^t \Theta(u)du
\end{align*}
Assume that $E\int_0^T \Theta^2(u)Z^2(u)du<\infty$. Let  $\tilde{\mathbb{P}}$ be the probability measure generated by the Radon-Nikodm derivative $Z(T)$. The the process $\tilde{B}(t)$ is a Brownian motion under $\tilde{\mathbb{P}}$
\end{theorem}
As we will see further on, continuous stochastic matrix product states allow us to map a Wiener process onto another, not necessarily Gaussian, process with correlated increment
\section{Stochastic Matrix Product states}
We finally introduce Stochastic Matrix Product States (sMPS) and a few of its continuum counterparts (csMPS).

Stochastic Product States were first introduced in (REF) as possible tool for studying  the partition function of spin models in statistical mechanics.
The stochastic product state was of the form,
\begin{equation}
\label{sMPSTemme}
|p_D\rangle =\sum_{i_1,...,i_N}\langle L|B_{(i_1)}^1...B_{(i_N)}^N|R\rangle|i_1...i_N\rangle
\end{equation}
with the matrices $B_{i_k}$ and vectors $\langle L|$, $\langle R|$ element-wise positive. Additionally the transfer matrices of this form $T[k]=\sum_j B_{(j)}[k]$ is a stochastic matrix.

We argue that this form is incomplete. In the last paragraph of this section, we show that for these type of states can be mapped onto a Markovian process. On the other hand, the construction proposed breaks down for the more general form, we give below.

Similarly to Matrix Product Operators (MPO), we derive Stochastic Matrix Product States from a purification method.

Consider a process $(X_j)_{j=1}^N$, $X_j\in \{y_1,...,y_d|y_j\in\mathbb{R}\}$, with joint probability $P(X_1=x_1,\dots,P_N=x_N)=p(x_1,\dots,x_n)$
Define the pure state $|\psi\rangle$,
$$|\psi\rangle =\sum_{i_1,\dots,i_N}\sqrt{p(x_1,\dots,x_N)}|i_1,\dots\rangle_A |i_1,\dots\rangle_B $$
The MPS-representation $\{B^{(i_j)}\}$ of this state can be derived. By tracing out the ancillary $B$, we derive the form,
\begin{equation}
\label{sMPSQM}
P(X_1=x_1,\dots,X_N=x_N)=C^N\sum_{i_1,\dots,i_N}\langle L|B_{(i_1)}^1 \circ ...\circ B_{(i_N)}^N|R\rangle \delta(x_1,i_1)\dots\delta(x_N,i_N)
\end{equation}
where $B_{(i_k)}^1$ is completely positive.
For some normalization constant $C$. In this case for the boundary operators, $|L\rangle =\hat{L} \otimes \mathbb{1}|I\rangle$, $|R\rangle =\hat{R}\otimes \mathbb{1}|I\rangle$, it should be taken that $\hat{L},\hat{R}\geq 0$. In contrast to the previous form the transfer matrix is a trace preserving completely positive operator $\Gamma[k][.]=\sum_j B_{(i_k)}^k [.]$. 

\paragraph{continuous-time sMPS}
Setting the boundaries first aside see that this form is exactly equivalent to the Quantum Measurements described earlier. We can make this more explicit with the gauge transformation $A^{(i)}\to \tilde{A}^{(i)}=X^{-1/2}\tilde{A}^{(i)} X^{1/2}$, $\rho \to \tilde{\rho}=X^{1/2}\rho X^{1/2}$.

Additionally, when rewriting,

$$P(X_1=x_1,\dots,P_N=x_N)=Z_N(x_1,...,x_N) \left(\frac{1}{d}\right)^N$$

We see that the sMPS functions behave as a change of measure.
Let us keep this in mind and rewrite $Z_N$ as a solution of a discrete stochastic differential equation for the case $d=2$ and $X_j=\pm1$. So $Z_N=\mathbb{1}+\sum_{j=1}^N \Delta Z_n$ and,
$$\Delta Z_n = Z_{n-1} \frac{\left(\Gamma[n]-\mathbb{1}\right)}{\Delta n}\Delta n+Z_{n-1}\sqrt{\Delta_n}\left(A^{(-1)}[n]\otimes \overline{A}^{(-1)}[n]-A^{(1)}[n]\otimes \overline{A}^{(1)}[n]\right) \Delta B_n$$
where $\Delta B_n=\delta(x_n,+1)-\delta(x_n,-1)$.  By taking the renormalization,
 $$A^{(-1)}=\mathbb{1}+\Delta_n Q+\sqrt{\Delta_n}R,~~ A^{(+1)}=\mathbb{1}+\Delta_n Q-\sqrt{\Delta_n}R$$
with $Q=iH+\frac{1}{2}R^\dagger R$, and taking limit $\Delta_n \to 0$, this leads to the form,
\begin{equation}
\label{csMPS}
Z(t)=Z(t)\mathcal{L}_0[.]dt +Z(t)\left(R[.]+[.]R^\dagger\right) dB_t
\end{equation}
derived in the Quantum Measurement section. Hence, the Radon-Nikodym satisfying is a continu version the sMPS.

Clearly other time-continuous versions are possible by considering a different product measure. Let $X_j\in\{0,1\}$ and $p(X_j=0)=\exp(-\lambda \epsilon)$, $p(X_j=1)=\exp(-\lambda \epsilon)\epsilon \lambda$.
Then let us write,
$$\Delta Z_n = Z_{n-1} \frac{A^{(0)}[n]\otimes \overline{A}^{(0)}[n]-\mathbb{1}}{\epsilon}\Delta n+Z_{n-1}\left(A^{(1)}[n]\otimes \overline{A}^{(1)}[n]-\mathbb{1}\right) \Delta N(n)$$
Here, we chose $\Delta n=\epsilon\delta(x_n=0)$and $\Delta N(n)=\delta(x_j=1)$. By taking $A^{(0)}=\mathbb{1}+i\epsilon H-\epsilon \frac{1}{2}\mu $ and $A^{(1)}=\sqrt{\mu} U$, for some unitary $U$, $U^\dagger U=\mathbb{1}$ and $\mu>0$.
Computing the generator $\mathcal{L}(\lambda)$ of the characteristic function yield,
$$\mathcal{L}(\lambda)=i[H,.]+\mu\left(\exp(i\lambda)U[.]U^\dagger-[.]\right)$$
In the bound dimension one case, this is nothing but a Poisson process. For higher dimension, this corresponds to a photon-counting process from the quantum measurement point of view.
We do not continue on this.

\paragraph{sMPS as fixed points of Metropolis sampling}

Let us rephrase the original Metropolis Monte Carlo Algorithm in a quantum measurement based language in a some state manifold $\mathcal{M}_{\mbox{{\scriptsize states}}}$. This allows us a larger variety of input states, besides product state, and operations, thus extending the markovian character of the algorithm to non-markovianity. 

The whole procedure can be resumed as follows,
starting from some input state $\rho (0)$
\begin{enumerate}
\item at step $t$, pick a (local) operation $O_j[.]$
\item accept the measurement with probability $\operatorname{Tr}\left(O_j^*[\mathbb{1}]\rho (t-1)\right)$
\item if accepted, change the state to 
$$\rho(t-1)\to \rho(t)  =\min\limits_{\sigma \in \mathcal{M}_{\mbox{{\scriptsize states}}}}\left\|\sigma -\frac{O_j[\rho (t-1) ]}{\operatorname{Tr}\left(O_j^*[\mathbb{1}]\rho (t-1)\right) }\right\|_2$$
if necessary, project the state back onto the manifold.
\item go back to step 1
\end{enumerate}
Given some observable $M$, for large enough time $T$, the expectation $E(M)(T)$, 
$$E(M)(T)= \frac{1}{T} \sum_{t=1}^T \operatorname{Tr} \left(|\psi (t)\rangle \langle \psi(t)| M \right)$$
converges to the equilibrium expectation, $\lim_{T\to \infty} E(M)(T)=\operatorname{Tr} \left(\rho_{\infty} M \right)$, with,
$$\rho_{\infty}=\sum_i\lim_{T\to \infty} \operatorname{Tr}\left(\rho_0 \Gamma^T [A^{i}A^{i\dagger}]\right)|i\rangle \langle i|,~~\Gamma[.]=\sum_j A_j [.] A_j^\dagger $$

As an example, consider a classical 1-dimensional spin chain, described by the Hamiltonian,
$$H=\sum_j s_i s_{i+1}$$
Choose $\mathcal{M}_{\mbox{{\scriptsize states}}}$ first to be the manifold of product states, 
$$\mathcal{M}_{\mbox{{\scriptsize states}}}=\left\{|s_1,...,s_N\rangle | s_i=\pm 1 \right\}$$
By taking the measurements $O(s_j)=\frac{\exp(\beta H) |s_j\rangle \langle -s_j| \exp(-\beta H)|}{\sum_{s_j}\exp(\beta H)}$, and defining the operations,
and operations, 
$$\tilde{O}(s_j)[.]=\sum_s\langle s_{j-1},s_j,s_{j+1}|O|s_{j-1},s_j,s_{j+1}\rangle \langle s_{j-1},s_j,s_{j+1}| [.] |s_{j-1},s_j,s_{j+1}\rangle | s_{j-1},s_j,s_{j+1}\rangle \langle s_{j-1},s_j,s_{j+1}|$$
the algorithm describes precisely a Glauber dynamics \cite{Metropolis}, which converges to the equilibrium state,
$$\rho_{\beta}=\sum_{s_1,...,s_N}\langle L| B(s_1)...B(s_N)|R\rangle |s_1,...,s_N\rangle\langle s_1,...,s_N|,~~B(s_j)=\left(\begin{array}{cc}
\exp(2\beta s_j)& 0\\
0 & \exp(-2\beta s_j)
\end{array}\right)$$

Let us now extend $\mathcal{M}_{\mbox{{\scriptsize states}}}$, to the mpo-manifold of bound dimension 2,
$$\mathcal{M}_{\mbox{{\scriptsize states}}}=\left\{\rho \left\{M[s_j]\right\}\rangle  | \operatorname{dim}M[s_j]=2 \right\},~~\rho\left\{M[s_j]\right\}=\sum_{s_1,...,s_N}\langle L| M(s_1)...M(s_N)|R\rangle |s_1,...,s_N\rangle\langle s_1,...,s_N|$$

Clearly $\rho_{\beta}$ is a fixed point, in the sense that,
$$\rho_{\infty}=\sum_j\operatorname{Tr}\left(\rho_\beta \Gamma^T [O^{j*}[\mathbb{1}]]\right)|i\rangle \langle i|,~~\forall T$$
\paragraph{Difference between the representations}
Equation (\ref{sMPSTemme}) can be mapped to (\ref{sMPSQM}), by taking, $$A^{(i)}_{y,z}[k]=\sqrt{B^{(i)}_{y,z}}[k]|y\rangle\langle z|$$
Even more, in the following construction we show that we can find a Markov process that is equivalent to the ones described by (\ref{sMPSTemme}) on the level of the MPS.
As noted before the transfer matrix of (\ref{sMPSTemme}) is a stochastic matrix. Consequently, this corresponds to the completely positive operator $\Gamma[]=\sum_{xy}T_{x,y}|x\rangle\langle y|$ which is the transfer matrix in (\ref{sMPSQM}).
Consider first a Markov process $(Y_j)$ with transfer matrix $T$. The joint probability distribution is easily written in smps-notation by taking $A^{(x_j)}[k]=T|j\rangle \langle j|$.
Take now the process $(X_j)$ with smps of the form  (\ref{sMPSTemme}). Consider its purification to MPS. Next block the sites over a length $L$ so that the number matrices $M^{(i_1,\dots,i_L)}=A^{(i_1)}\dots A^{(i_L)}$ exceeds the Kraus-rank of the completely positive operator (CP). Next apply a change of basis in the virtual. As it is known, for any CP-map Kraus-operators are equivalent by some unitary $U$, 
$$\Gamma[.]=\sum_j A^j[.]A^{j\dagger}=\sum_j B^j[.]B^{j\dagger}\leftrightarrow A^j=\sum_i U_{ji}B_j$$
Since we have blocked to full rank, and since the transfer matrix over this length is of the form $\Gamma[L][.]=\sum_{x,y}(T^L)_{x,y}|x\rangle\langle y|.|y\rangle\langle x|$. Take the unitary so that $M^{(i_1,\dots,i_L)}\to \lambda(x) \tilde{A}^{(x)}=(T^L)|x\rangle\langle x|$. Consider the equivalent class $[x]=\{[i_1,...,i_L]|M^{(i_1,\dots,i_L)}\propto  \tilde{A}^{(x)}\}$. The new process $Y_j=\{[x]\}$ is thus a Markovian process over a rescaled time.

In some sense the MPS describes the time evolution of a cavity couple an electromagnetic field. At each time, some detection procedure is applied, for example Homodyne detection, fixing the real basis of the MPS, i.e. photon basis at each time slice of the e.m.-field. The change of processes is therefore nothing but a change to another commutant of the *-algebra of observables.

\subsection{Master Equation for memory processes}
Let $(X_j)_{j\geq 0}$, be a process whose joint probability distribution is described by $\mathbb{P}$.
A topic of broad interest is the time evolution the marginal discribution of $X_N$. For any distribution, this time evolution of $P(X_N=x_N)$ with initial condition $P(X_0=x_0)$ is given by,
\begin{equation}
\label{memory}
P(X_N=x_N)=\sum_{x_1,...,x_{N-1}} P(X_N=x_n|X_{N-1}=x_{n-1},...,X_0=x_0)...P(X_1=x_1|x_0)
\end{equation}
For a Markov process with transition matrix $P(X_j=x_j|X_{j-1}=x_{j-1})=T_{x_j,x_{j-1}}$, this evolution can be reduced severely to the different forms,
$$P(X_N=x_N)=\sum_{x_{N-1}}T_{x_N,x_{N-1}}P(X_{N-1}=x_{N-1})$$
For $T=\mathbb{1}+\epsilon G$, with $G$ understood as a generator of a time-continuous stochastic matrix, we can rewrite in a differential form,
$$\frac{dP(k)}{dt}=\sum_l \delta(k,l)\frac{P(X_N=k)-P(X_{N-1}=l)}{\epsilon}=\sum_{l\not=k}\left[G_{k,l}P(X_{N-1}=l)-G_{l,k}P(X_{N-1}=k)\right]$$
This is known as the Master equation.

As we can see, the number of parameters of equation (\ref{memory}) scales exponential as we try to find the time-evolution of a marginal $P(X_N=x_N)$. In a first approach, as it is done in practice, we the transition probability can be bounded up to some fixed time $k$. However, similarly to the Markovian case, this can be written in a simple sMPS form of bound dimension $2k$.

\begin{example}
Let $k<\infty$ and $\forall n$, 
\begin{align*}
P(X_n=i_n|X_{n-1}=i_{n_1},&\dots,X_{n-k}=i_{n-k},X_{n-k-1}=x_{n-k-1},\dots,X_{}=x_{1},X_0=x_0)\\
&=P(X_n=i_n|X_{n-1}=j_{n_1},\dots,X_{n-k}=x_k)=T_{i_{n-k},...,i_{n-1},i_n}
\end{align*}
Then, the sMPS representation is given by,
$$A^{i_n}_{y_1,\dots,y_{k},z_1,\dots,z_{k}}=T_{y_1,\dots,y_k,z_{k}}\delta_{i_n=z_k}\delta_{y_{2}=z_1}\dots \delta_{y_{k}=z_{k-1}}$$
and with  boundaries, 
$$|X\rangle =|I\rangle,~~\langle \rho|_x=P(X_0=x)$$
\end{example}
As understand from the point of view of Quantum Measurement theory, the sMPS-formalism, describes smoother conditional probability. From the condensed matter point of view, similarly to matrix product states. Such states describe processes whose almost, but not quite joint probability factorizes and vice-versa \cite{BRANDAO},
$$\|p(X_1,...,X_k, X_{k+N},...,X_{k+l+N})-p(X_{k+N},...,X_{k+l+N})\|_1 \leq C_2 \exp(-C_1 N )$$

Thus, we can write for a compationally efficient parametrzation for such memory processes.
So given a process $(X_j)_{j\geq 0}$, described by the smps, 
$$P(X_1=x_1,\dots,X_N=x_N)=\operatorname{Tr}\left(\rho \hat{Z}_N{x_1,...,x_N}\right)$$
with $Z_0=\mathbb{1}$ and,
$$\Delta Z_N(x_1,...,x_N) =\sum_j \hat{Z}_{n-1}(x_1,...,x_{N-1})\left(A^{j}\otimes \overline{A}^j-\mathbb{1}\right) \delta(j,x_N)$$
Then,
$$P(X_N=k)=\operatorname{Tr}\left(\sum_{l}\hat{T}_{k,l}\hat{P}_l\rho\right) $$
for which,
$$\sum_l\hat{T}_{k,l}\hat{P}_l=\sum_l(\overline{A}^{(k)}\overline\otimes {A}^{(k)})(\overline{A}^{(l)}\overline\otimes {A}^{(l)})\Gamma^{*N-2}$$
In the time-continuous limit this corresponds to,
$$\frac{d}{dt}P(X_t=k)=\operatorname{Tr}\left(\sum_{l}\hat{G}_{k,l}\hat{P}_l\rho\right) $$
and again,
$$\sum_l\hat{G}_{k,l}\hat{P}_l=\sum_l S^k\exp(tL)$$
with $S^{(k)}[.]=Q^{(k)}[.]+[.]Q^{(k)\dagger}$,  $S^{(k)}[.]=R^{(k)\dagger}[.]R^{(k)}$ or sum of both
\begin{example}[Non-Markovian Birth-Death processes]
Let $X_j\in \mathbb{N}$, and 
$$Q^{(n)}=-\frac{1}{2}\hat{G}_{n,n}\otimes|n\rangle \langle n|,~~R_{+1}^{(n)}=\hat{G}_{n,n+1}|n\rangle \langle n+1|,~~R_{-1}^{(n)}=\hat{G}_{n,n-1}|n\rangle \langle n-1|$$
with $0=\hat{G}_{n,n}+\hat{G}_{n,n}^\dagger +\hat{G}_{n+1,n}\hat{G}_{n+1,n}^\dagger+\hat{G}_{n+1,n} \hat{G}_{n+1,n}^\dagger$.

Let $S^{(n)}$ be,
$$S^{(n)}=Q^{(n)}[.]+[.]Q^{(n)\dagger}+  R_{+1}^{(n)\dagger}[.]R_{+1}^{(n)}+ R_{-1}^{(n)\dagger}[.]R_{-1}^{(n)}$$

For choice of boundary $\rho=\sum_n \lambda_n |n\rangle \rangle n|$ and $\operatorname{dim}(\hat{G}_{n,m})=1$, this reduces to the birth-and-death process.
\end{example}

\subsection{Description of Non-Markovian Quantum Dynamics}
Consider some two level system, representing the two lowest energy states of a larger n-level system weakly couple to an infinite environment. For low temperatures and weak-coupling, the density matrix $\rho(t)o$ evolves under some Markovian dynamics, $\Gamma_t =\exp(tL)$,
$$\rho_0 \to \rho(t)=\exp(tL)[\rho_0]$$
Let us now increase the temperature $T>0$, but keep a weak coupling with environment, while observing the evolution of the two level system as being part of the initial Gibbs-state $\rho_T$. This two level system, is a part of some subspace of the evolving state $\rho_T(t)$. 

Define the super-operator $M_{ij}[.]:\mathcal{M}_n\to \mathcal{M}_n$.

\begin{definition}
\label{defMij}
Let $M_{ij}[.]:\mathcal{M}_n\to \mathcal{M}_n$, so that
so that $\forall x_{i}\overline{x}_j$, $\sum_{i,j}x_i \overline{x}_j M_{ij}[.]$ is completely positive and $\sum_j M_{jj}[\mathbb{1}]=\mathbb{1}$. 
\end{definition}

Notice then that $\forall \sigma \in \mathbb{M}_n$, the matrix $\rho=\operatorname{Tr}\left(\sigma M_{ij}[\mathbb{1}]\right)|i\rangle \langle j|$ is a density matrix. On the other hand, any projection of a higher-level system to a lower level subsystem, has to satisfy these condition. Therefore, the definition (\ref{defMij}), should be seen as the necessary and sufficient condition for any projector of a density matrix onto a lower dimensional density.

With this operator, we can thus extract the density matrix of the two-level system at each time t,
\begin{equation}
\label{MasterQ}
\rho_T(0)\to \rho(t)=\operatorname{Tr}\left(\exp(t L)[\rho_T] M_{ij}[\mathbb{1}] \right)|i\rangle \langle j|
\end{equation}

We can see that the canonical form of $M_{ij}[.]$ is of the form,

$$M_{ij}[.]=\sum_k  A^{i,k}[.]A^{j,k \dagger},~~\sum_{j,k}  A^{j,k}[\mathbb{1}]A^{j,k \dagger}=\mathbb{1}$$

Indeed, the idea is similar to the Choi-Jamiolkowski Isomorphism for deriving the Kraus-representation of completely positve operators.
Define the matrix,
$$=\sum_{i,j,\alpha,\beta} |i\rangle \langle j|\otimes M_{ij}[|\alpha\rangle\langle \beta|]\otimes |\alpha\rangle\langle \beta|$$
By definition of $M_{ij}$, $\mathcal{C}(M)\geq 0$, whence it can be decomposed as,
\begin{align*}
\mathcal{C}(M)&=\sum_k |X_k\rangle \langle X_k|=\sum_k |i\rangle \langle j|\otimes X_{i,k}|I\rangle \langle I| X_{j,k}^\dagger
\end{align*}
where we have written, $|I\rangle =\sum_j |jj\rangle$.
From this the proof follows.

We see that starting from a quantum Markovian dynamics, the effective evolution of the two-level system given in equation (\ref{MasterQ}) is non-Markovian.

Let us derive this evolution from a Matrix Product State approach. Once this is done, we can conclude using the area law result that any non-Markovian dynamics, for which, 
$$\|\rho(X_1,...,X_k, X_{k+N},...,X_{k+l+N})-\rho(X_{k+N},...,X_{k+l+N})\|_2 \leq C_2 \exp(-C_1 N )$$
can be efficiently approximated using equation using MPO and the time-continuous dynamics is given by equation (\ref{MasterQ}).

Denote $\rho(X_1,\dots, X_N,\dots)$ the joint density matrix of process describing the evolution of some sub-level system, in the course of our example a 2-level system.
Consider the purification of the process, 
$$|\psi\rangle =\sum_{i_1,\dots, i_N,\dots}\sqrt{\rho(X_1,\dots, X_N,\dots)}|i_1,\dots, i_N,\dots\rangle |i_1,\dots, i_N,\dots\rangle$$
Again, considering the MPS-representation of the state, and tracing out the ancillary, we acquire the Matrix Product Operator representation of the process,
$$\rho(X_1,\dots, X_N,\dots)=\sum_{i_1,j_1,\dots, i_N,j_N,\dots}\operatorname{Tr}\left(\rho M_{i_1,j_1}[.]\circ \dots \circ  M_{i_N,j_N}[.]\circ \dots \mathbb{1} \right)$$
with,
$$M_{i_k,j_k}[.]=\sum_{l}A^{il}[.]A^{il\dagger}$$
This yields us the evolution of the marginals,
Then,
$$\rho(N)=\sum_{i,j}\operatorname{Tr}\left(\hat{\Gamma}[\hat{\rho}(N-1)_{ij}]\rho\right)|i\rangle\langle j| $$
for which,
$$\hat{\Gamma}[\rho(N-1)]=M_{ij}^*\circ \Gamma^{*N-1})$$
In the time-continuous limit this corresponds to,
$$\frac{d}{dt}\rho(t)=\operatorname{Tr}\left(\hat{G}[\hat{\rho}(t)_{ij}]\rho\right)|i\rangle\langle j|$$
and again,
$$\hat{G}[\rho(t)]=S_{ij}^*\circ \exp(tL)$$

with $S^{(ij)}[.]=Q^{(i)}[.]+[.]Q^{(j)\dagger}$ or $S^{(ij)}[.]=R^{(i)\dagger}[.]R^{(j)}$, and,
$$L[.]=\sum_j Q^{(j)}[.]+[.]Q^{(j)}+R^{(j)}[.]R^{(j)},~~\sum_{j} Q^{(j)}+Q^{(j)\dagger}+R^{(j)\dagger}R^{(j)}=0$$

\section{Miscellaneous}

\subsection*{A Complete Market with Finitely Correlated Increments of the Logarithm of the Stocks}
As we can see, from the form of the Radon-Nikodym derivatives described in this paper, there is for all increments some correlation with future and past increments. One of the goals of Financial Mathematics is to study the pricing of contracts so that both clients and agents cannot take advantage of each other. This theory is also known as arbitrage-pricing theory. Interestingly, our formalism allow us to define a market with a certain bias in the evolution of stocks. Yet, we see that among these markets, fairness, in the sense of "no-arbitrage", can still be found.

Idea, consider the projector onto the identity. With a 
\subsubsection{A Fast Introduction to Neutral-Pricing}
This introduction is meant as a presentation of the ideas of arbitrage pricing and contains many holes. For a full introduction we refer to \cite{SHREVEI}, \cite{SHREVEII}. 
Let $(\Omega,\mathcal{F},\mathbb{P})$ be a probability and $B_t$ a brownian motion. Denote $\mathcal{F}(t)$ the filtration wrt to the Brownian motion. 
In the most simple market model to different actions are possible at each time $t$. We can either purchase some stock $S(t)$ or invest at some interest rate $r(t)$ in the money market. An option is a contract with a payoff at some later time $T$ that depends on the stock $S(t)$ at different values $0<t\leq T$. 
The idea of arbitrage pricing approach is to find a portofolio $X(t)$ that replicate the options by only using the two actions described above.

Assume the stock whose price at time $t$ is given by $S(t)$ satisfies the stochastic differential equation,
\begin{equation}
\label{Stock}
dS(t)=\alpha S(t)dt+\sigma S(t)dB_t
\end{equation}

Define the discount process $D(t)$,
$$D(t)=\exp\left(-\int_0^t r(s)ds\right)$$
The randomness of the stock is described by some measure $\mathbb{P}$. However, as we will see there is a nice trick that allow us to find the replicating strategy by changing to a new measure, called the risk-neutral measure. 

\begin{definition}
A probability measure is said to be risk-neutral if,
\begin{enumerate}[(i)]
\item $\mathbb{P}$ and $\tilde{\mathbb{P}}$ are equivalent, and
\item under $\tilde{\mathbb{P}}$, the discounted stock price $D(t)S(t)$ is a martingale
\end{enumerate}
\end{definition}
The existence of the risk neutral measure plays a central in risk-neutral pricing as it allows to find a strategy to hedge any derivative security.

Denote $X(t)$ the value of the portofolio. The idea is to hedge the derivative security by investing at each time in $\Delta(t)$ stocks $S(t)$ and the difference in the money market with interest rate $r$.
The differential of $X(t)$ is therefore given by,
$$dX(t)=\Delta(t)dS(t)+r(X(t)-\Delta(t) S(t))dt=rX(t)dt+\Delta(t)d(D(t)S(t))$$
We see then that the differential of the discounted portofolio,
$$d(D(t)X(t))=\Delta(t)d(D(t)S(t))$$
The martingale representation theorem allows us the write any martingale defined on a filtration of a Brownian motion as a stochastic integral. In the case of $D(t)S(t)$, we have $D(t)S(t)=\int_0^t \sigma D(s)S(s)dW_s$.
The first fundamental theorem of asset pricing asserts that if a market model has a risk-neutral-pricing formula, then it does not admit arbitrage. Arbitrage is a trading strategy that begins with nothing, has probability of losing money, and a strictly positive probability of making money. 

Let $V(T)$ be the payoff of the option at time $T$. The payoffs at a time $t<T$ are called derivative securities.
Therefore, we see that our model
So given some contract given by the derivative security $V(T)$. Then since $D(T)S(T)$ is a martingale under $\tilde{\mathbb{P}}$,then so is $D(t)V(t)$.
We can then use the martingale representation theorem to write,
$$D(t)V(t)=V(0)+\int_0^t \Gamma dB_s$$
Our portofolio will replicate the price of the derivative security, if $\Delta(t)=\frac{\Gamma(t)}{\sigma D(t)S(t)}$ and $X(0)=V(0)$.
Such a market wherein every derivative security can be hedged is called a complete market.

\subsubsection{The Correlated Model}

Clearly the form of the SDE (\ref{mainSDE}) of the stochastic matrix product state reminds a lot of geometric brownian motion. So why not define a model, where the stocks $S(t)$ are described by such SDE.
According to the arbitrage pricing, our market model is then complete if $D(t)S(t)$ is a martingale under some measure $\tilde{\mathbb{P}}$ whose Radon-Nikodym derivative is also of the form (\ref{mainSDE}). Unfortunately, it turns out that we cannot consider both $S(t)$ and  $\tilde{\mathbb{P}}$ of the form (\ref{mainSDE}) at the same time.
We discuss both cases separately and derive the condition for the completeness of the market.

\subsubsection{Case 1: Change in Evolution Stock}

The idea of model is then to redefine the time evolution of the stock $S(t)=\operatorname{Tr}\left(\rho \hat{S}(t) X\right)$, with,
\begin{equation}
\label{Stock}
 d\hat{S}(t)=\hat{S}(t)\left(\mathcal{L}_0[.]dt +\alpha dt +\sigma\left(R[.]+[.]R^\dagger\right) dB_t \right)
\end{equation}
The discounted process $D_t S_t$ will have a drift terms $\alpha-r$. This drift term can be taken care of using Girsanov's theorem.

\begin{theorem}
Let,
\begin{equation}
\label{martingalecond1}
\mathcal{L}_0[X]dt +(\alpha -r) \left(R[X]+[X]R^\dagger\right)=0
\end{equation}
with $\mathcal{L}_0[.]$ of the form (\ref{Lindblad}).
Let $S(t)$ be a process satisfying equation (\ref{Stock}). For the Radon-Nikodym derivative $Z_T$,
$$Z(T)=\exp\left(-\frac{\alpha -r }{\sigma}B(T)-\frac{1}{2}\left(\frac{\alpha -r }{\sigma}\right)^2T\right)$$
Let $\tilde{\mathbb{P}}$ be the measure generated by $Z_T$ is a risk-neutral measure. Then if,
$$E\left(Z_T\operatorname{Tr}\left(\rho \hat{S}(t)\left(R[X]+[X]R^\dagger \right)\right)^2\right)<\infty$$
 $\tilde{\mathbb{P}}$ is the risk-neutral measure
\end{theorem}
\begin{proof}
First, we use Girsanov's theorem to shift the drift of $D(t)S(t)$ by going over to the Brownian process under $\tilde{\mathbb{P}}$,
$$d\tilde{B}(t)=dB(t)+\frac{(\alpha-r)}{\sigma}dt$$
Under the new measure $\tilde{\mathbb{P}}$, $D(t)S(t)$ is of the form,
$$dD(t)S(t)=D(t)\operatorname{Tr}\left[\rho \hat{S}(t)\left(\mathcal{L}_0[.]dt+(\alpha -r) \left(R[.]+[.]R^\dagger\right)dt+ \sigma\left(R[.]+[.]R^\dagger\right) d\tilde{B}_t\right)X\right]$$
Under the condition of the theorem, the $dt$ term disappears and $D(t)S(t)$ is a martingale under $\mathbb{\tilde{P}}$.

\end{proof}
By the theorem, proven above our market model is therefore a complete market.

In condensed matter, the thermodynamic limit is often of interest. In this case, this corresponds to the limit $T\to \infty$. Something interesting happens in this case. Indeed for any finite $T<\infty$, we see that the boundaries have to correspond with the left and right zero-eigenvector of $\mathcal{L}_0$. In the case, the condition derived in previous theorem reduces to,
$$\left(R[X]+[X]R^\dagger \right)=0$$
Hence the boundary is an eigenvector of both $R[.]+[.]R^\dagger$ and $\mathcal{L}_0$. This means that the bound dimension of $R$ reduces to $1$, and $Z_t=1$. So if we consider some option that has a payoff at a finite time, then we get the usual stock again.
$$\lim_{T\to \infty}\tilde{S}(t)=S(t)$$
with,
$$dS(t)=\alpha dt + \sigma dB(t)$$

We do not continue further on this case.
\subsubsection{Case 2: Change of Measure}
The model here considers the new measure $\tilde{\mathbb{P}}$ generated by the Radon-Nikodym derivative $Z_T=\operatorname{Tr}\left(\rho \hat{Z}_T X\right)$,
\begin{equation}
\label{Radon Mod2}
 d\hat{Z}(t)=\hat{Z}(t)\left(\mathcal{L}_0[.]dt +\left(R[.]+[.]R^\dagger\right) dB_t + m dB_t\right)
\end{equation}
And the stocks $S(t)$ satisfy the usual evolution,
$$dS(t)=\alpha dt + \sigma dB(t)$$
The increments $\Delta B_n$ are now possibly correlated under this measure.
First, let us make sure that  the discounted stock $D(t)S(t)$ is a martingale again.
Assume further on that $Z(T)>0$.
\begin{theorem}
For, 
\begin{equation}
\label{martingalecond}
\mathcal{L}_0[X]+(\alpha-r+\sigma m)X+  \sigma\left(R[X]+[X]R^\dagger \right)=0
\end{equation}
with $\mathcal{L}_0[.]$ of the form (\ref{Lindblad}).
The measure $\tilde{\mathbb{P}}$ defined under the Radon-Nikodym derivative (\ref{Radon Mod2}), then,
 $$d(D(t)S(t)Z_t)=\sigma D(t)S(t)\mathcal{Z}_t dB_t$$
 with,
 $$\mathcal{Z}(t)=\operatorname{Tr}\left(\rho \hat{Z}(t)  \left(R[X]+[X]R^\dagger+(m+\sigma)X\right)\right)$$
 If,
 $$E(S(t)^2\mathcal{Z}(t)^2)<\infty$$

then $\tilde{\mathbb{P}}$ is a risk neutral measure
\end{theorem}
\begin{proof}
The computation is similar as the one above.
Using ito calculus, we can see that,
$$E(Z_T D(t)S(t)|\mathcal{F}_s)=Z(s)D(s)S(s)$$
where $Z(s)=E(Z_T |\mathcal{F}_s)$ is the Radon-Nikodym process.
From this we indeed see that $D(s)S(s)$ is a martingale under $\tilde{\mathbb{P}}$.
\end{proof}

Additionally, we need to show that for this model every derivative security can be hedged. Indeed, in the introduction the martingale representation theorem was used. However, $B(t)$ is not a Brownian motion anymore and not even a martingale under $\tilde{\mathbb{P}}$. Yet, we need to define a method.
Notice first the following,
\begin{Lemma}
If $M_t$ is a martingale under $\tilde{\mathbb{P}}$, then $Z_t M_t$ is a martingale under $\mathbb{P}$.
\end{Lemma}
\begin{proof}
Since $M(t)$ is $\mathcal{F}_t$-measurable then,
$$M_s=\tilde{E}(M_t|\mathcal{F}_s)=\frac{1}{Z_s}E(Z_t M_t |\mathcal{F}_s)$$
from which the claim follows.
\end{proof}
We can then show using the theorem that for the portofolio process $X(t)$,
$$d(Z(t)S(t)X(t))=\sigma \Delta(t)S(t)\mathcal{Z}(t) X(t) dB_t$$

Defining the martingale process for the derivative security, $D(t)V(t)=\tilde{E}(D(T)V(T)|\mathcal{F}_s)$, the process $Z(t)D(t)V(t)$ is a martingale under $\mathbb{P}$. The martingale representation theorem can be used again,
$$Z(T)D(T)V(T)=V(0)+\int_0^t \Gamma(s) dB_s$$
And we can set $\Delta(t)=\frac{\Gamma_t}{\sigma S(t)\mathcal{Z}_t}$.
 Since the portofolio can only replicated by a unique strategy, this market model is complete.

The thermodynamic limit $T\to \infty$ can be discussed again. Indeed, again, we need the additional condition $\mathcal{L}_0[X]=0$. the Radon-Nikodym process then reduces to,
$$\lim_{T\to \infty}Z(t)=\exp\left(-\frac{\alpha -r }{\sigma}B_t-\frac{1}{2}\left(\frac{\alpha -r }{\sigma}\right)^2t\right)$$
We summarize this in the corollary,
\begin{corollary}
In the limit $T\to \infty$,  a complete market model will always have in time uncorrelated increments of the logarithm of the stock.
\end{corollary}

\section{Conclusion}
In this paper, we discuss how Stochastic Matrix Product states can be used as a representation of the Radon-Nikodym derivative of processes defined on a filtration of another process. With this representation, we derive a simple equation for the time-evolution of the marginal distribution of the process. We showed how classical non-Markovian, classical and quantum processes can be embedded in a quantum Markovian dynamics. The properties of non-Markovian processes, such as ergodicity and mixing are therefore determined by the quantum dynamics.

\end{document}